\documentclass[preprint,12pt,fleqn]{elsarticle}

\usepackage{amssymb}
\usepackage{amsfonts}
\usepackage{amsthm}
\usepackage[english]{babel}
\usepackage{euscript}
\usepackage{bbm}
\usepackage{xfrac}
\usepackage{pb-diagram}

\usepackage{mathrsfs}
\usepackage{calrsfs}
\usepackage{amsbsy}
\usepackage{url}
\usepackage{hyperref}

\newcommand{\textfrac}[2]{\textstyle{\frac{#1}{#2}}}

\newcommand{\D}{\EuScript{D}}
\newcommand{\E}{\mathrm{E}}
\newcommand{\J}{\mathrm{J}}

\usepackage{amssymb}
\usepackage{amsmath}
\usepackage{amsfonts}
\usepackage[english]{babel}
\usepackage{bbm}
\usepackage{xfrac}
\usepackage{color}
\usepackage{url, hyperref}

\usepackage{mathrsfs}

\usepackage[all]{xy}
\SelectTips{cm}{}

\let\mathcal\mathscr

\newtheorem{theorem}{Theorem}[section]

\newtheorem{definition}{Definition}

\newtheorem{remark}{\textsc{Remark}}

\let \kappa=\varkappa
\let \phi=\varphi

\newcommand{\ldb}{[\![}
\newcommand{\rdb}{]\!]}

\newcommand*{\eval}[1]{\left.#1\right|}

\newcommand*{\Ev}{\mathbf{E}}

\DeclareMathOperator{\CDiff}{\mathcal{C}Diff}
\newcommand{\abs}[1]{\vert#1\vert}

\begin{document}

\begin{frontmatter}

  \title{Integrability structures of the $(2+1)$-dimensional Euler equation}
  \tnotetext[t1]{The work of I.S.K.  was partially supported by the RSF Grant
    25-71-20008}

\author{I.S. Krasil'shchik}%\corref{cor1}
\ead{josephkra@gmail.com}
\address{Trapeznikov Institute of Control   Sciences, 65 Profsoyuznaya street, 
\\
Moscow 117997, Russia}

\author{O.I. Morozov}   %\corref{cor1}
\ead{oimorozov@gmail.com}
\address{Trapeznikov Institute of Control   Sciences, 65 Profsoyuznaya street, 
\\
Moscow 117997, Russia}

\begin{abstract}
  We construct a local variational Poisson structure (a Hamiltonian operator)
  for the $(2+1)$-dimensional Euler equation in vorticity form.  The inverse
  defines a nonlocal symplectic structure for the equation. We describe the
  action of this operator on the infinitesimal contact symmetries in terms of
  of differential coverings over the Euler equation. Furthermore, we construct
  a nonlocal recursion operator for cosymmetries. Finally, we generalize the
  local variational Poisson structure for the Euler equation in vorticity form
  on a two-dimensional Riemannian manifold.
\end{abstract}


\end{frontmatter}

% \linenumbers
%=====================================================================
% \maketitle

%=====================================================================

\section{Introduction}

We consider the (2+1)-dimensional Euler equation in the vorticity
form~\cite[\S~10]{LandauLifshits6}
\begin{equation}
\Delta u_t = \J (u, \Delta u),
\label{Euler_eq}
\end{equation}
where $\Delta u = u_{xx} + u_{yy}$ and $\J(a,b) = a_x b_y-a_y b_x$.  This
equation is a cornerstone model in hydrodynamics and has been the focus of
extensive research; for more details, we refer to~\cite{ArnoldKhesin} and the
references therein, as well as to~\cite{KuznetsovMikhailov1980,Olver1982}.

To simplify computations in the subsequent sections, we perform the following
change of variables: we write Eq. \eqref{Euler_eq} as
$\tilde{u}_{\tilde{t}\tilde{x}\tilde{x}}+\tilde{u}_{\tilde{t}\tilde{y}\tilde{y}}
=
\tilde{u}_{\tilde{x}}\,(\tilde{u}_{\tilde{x}\tilde{x}\tilde{y}}+\tilde{u}_{\tilde{y}\tilde{y}\tilde{y}})
-\tilde{u}_{\tilde{y}}\,(\tilde{u}_{\tilde{x}\tilde{x}\tilde{x}}+\tilde{u}_{\tilde{x}\tilde{y}\tilde{y}})$
and then put $\tilde{t}=t$, $\tilde{x}=\frac{1}{2}\,(1+\mathrm{i})\,(x+y)$,
$\tilde{y} =-\frac{1}{2}\,(1-\mathrm{i})\,(x-y)$, and $\tilde{u}=u$, where
$\mathrm{i} =\sqrt{-1}$.  This yields the equation
\begin{equation}
\D u_t = \J (u, \D u),
\label{main_eq}
\end{equation}
where $\D = D_x\circ D_y$ is the composition of total derivatives.

In this paper, we construct a local variational Poisson structure (a
Hamiltonian operator) for Eq. \eqref{main_eq}. The inverse of this operator
defines a nonlocal symplectic structure for the equation. We describe the
action of this operator on the infinitesimal contact symmetries in terms of
pseudopotentials of differential coverings over the Euler
equation. Furthermore, we consruct a nonlocal recursion operator for
cosymmetries. Finally, we generalize the local variational Poisson structure
to the Euler equation in vorticity form on a two-dimensional Riemannian
manifold.

In what follows, we use the definitions and the techniques of 
\cite{KerstenKrasilshchikVerbovetskyVitolo2009,KrasilshchikVerbovetsky2011,KrasilshchikVerbovetskyVitolo2017}.

%=====================================================================

\section{Variational Poisson structures}
\label{sec:vari-bivect}

Consider an infinitely prolonged equation $\mathcal{E}\subset J^\infty(\pi)$,
where $\pi\colon E\to M$ is a vector bundle, see~\cite{VK1999}.  To simplify
exposition, we assume $M = \mathbb{R}^n$ and
$E = \mathbb{R}^m\times\mathbb{R}^n$ with the local coordinates
$x^1,\dots,x^n$ and $u^1,\dots,u^m$ in~$ \mathbb{R}^n$ and~$ \mathbb{R}^m$,
resp. Let~$\mathcal{E}$ be defined by the equation $F = 0$, where
$F = (F^1,\dots,F^m)$ is a vector-function on the space~$J^\infty(\pi)$. We
denote by~$P$ the $\mathcal{F}$-module of such functions,
where~$\mathcal{F} = C^\infty(J^\infty(\pi))$
(or~$C^\infty(\mathcal{E})$). Below,~$\kappa$ denotes the module of sections
of the pull-back~$\pi_\infty^*(\pi)$, where
$\pi_\infty\colon J^\infty(\pi)\to M$
and~$\hat{\bullet} = \hom_{\mathcal{F}}(\bullet, \Lambda_h^n)$; here
$\Lambda_h^n$ is the space of horizontal $n$-forms.

A $\mathcal{C}$-differential operator $\nabla\colon A\to B$, $A$ and~$B$ being
$\mathcal{F}$-modules, is an operator in total derivatives. The set of such
operators is denoted by $\CDiff(A, B)$. Such operators admit restriction
from~$J^\infty(\pi)$ to~$\mathcal{E}$. In particular~$\ell_F$, the
linearization of~$F\in P$, is a $\mathcal{C}$-differential operator and we set
$\ell_{\mathcal{E}} = \eval{\ell_F}_{\mathcal{E}}\colon\kappa\to P$. Similar
notation is used for the adjoint operator~$\ell_F^*$.

In what follows we assume that:
\begin{enumerate}
\item $\mathcal{E}$ is a smooth submanifold in~$J^\infty(\pi)$;
\item $\mathcal{E}$ surjectively projects by~$\pi$ onto~$J^0(\pi)$;
\item If an ``object'' (function, differential form,
  $\mathcal{C}$-differential operator, etc.) $O\in\mathcal{O}$, $\mathcal{O}$
  being some $\mathcal{F}$-module, admits a
  restriction~$\eval{O}_{\mathcal{E}}$ from~$J^\infty(\pi)$ to~$\mathcal{E}$
  and this restriction vanishes, then
  \begin{equation}
    \label{eq:1}
    O = \nabla(F), \qquad \nabla\in\CDiff(P, \mathcal{O});
  \end{equation}
\item The compatibility complex of~$\ell_{\mathcal{E}}$ is trivial.
\end{enumerate}

\begin{definition}
  Let~$\mathcal{E} = \{F = 0\}\subset J^\infty(\pi)$. Then
  $\mathcal{T}^*\mathcal{E} = \{F = 0,\ \ell_{\mathcal{E}}^*(p) = 0\}$ is
  called the cotangent equation of~$\mathcal{E}$, while
  $\mathrm{t}^*\colon \mathcal{T}^*\mathcal{E} \to \mathcal{E}$ is the
  cotangent covering. Here~$p = (p^1,\dots,p^m)$ is an odd variable of
  parity~$1$. Similarly,
  $\mathcal{T}\mathcal{E} = \{F = 0,\ \ell_{\mathcal{E}}(q) = 0\}$ is called
  the tangent equation of~$\mathcal{E}$, while
  $\mathrm{t}\colon \mathcal{T}\mathcal{E} \to \mathcal{E}$ is the tangent
  covering.
\end{definition}

\begin{remark}
  Due to oddness of the variables~$p_\sigma^j$ (or~$q_\sigma^j$), all the
  functions are polynomial in these variables, while their multiplication is
  anti-symmetric: $fg = (-1)^{\abs{f}\,\abs{g}}gf$, where~$\abs{\cdot}$
  denotes parity. This fact is used essentially in the proof of
  Theorems~\ref{Hamiltonian_operator_theorem} and~\ref{thm:hamilt-oper-euler}.
\end{remark}

\begin{definition}
  A $\mathcal{C}$-differential operator $H\colon \hat{P}\to\kappa$ is a
  (variational) bivector on~$\mathcal{E}$ if
  \begin{gather}
    \label{eq:2}
    \ell_{\mathcal{E}}(H) = 0,\\
    \label{eq:3}
    (\ell_{\mathcal{E}}\circ H)^* \equiv H^*\circ\ell_{\mathcal{E}}^* =
    \ell_{\mathcal{E}}\circ H.
  \end{gather} 
\end{definition}

\begin{remark}
  In coordinates, $H$ is of the form
  $H = (\sum_\sigma h_{ij}^\sigma D_\sigma)\in\CDiff(\hat{P}, \kappa)$, where
  $D_\sigma$ denotes the composition of total derivatives that corresponds to
  the multi-index~$\sigma$. Consider the function
  $H_u = (\sum_{\sigma,j} h_{ij}^\sigma p_\sigma^j)$
  on~$\mathcal{T}^*\mathcal{E}$. The expression~$\ell_{\mathcal{E}}(H)$ in
  Eq.~\eqref{eq:2} is to be understood as~$\ell_{\mathcal{E}}(H_u)$.
\hfill $\diamond$
\end{remark}

Denote by
$\Ev_\xi = \sum_{j,\sigma} D_\sigma(\xi^j)\frac{\partial}{\partial
  u_\sigma^j}$ the evolutionary field with the generating
section~$\xi\in\kappa$.

\begin{theorem}[cf.\ \cite{KrasilshchikVerbovetsky2011, KrasilshchikVerbovetskyVitolo2017}]
  Let $H$ be a bivector on~$\mathcal{E}$. Then there exists an element~$H_p$
  of parity~$2$ such that the evolutionary derivation~$\Ev_{\phi(H)}$ of
  parity~$1$ with the generating section~$\phi(H) = (H_u, H_p)$ is a symmetry
  of~$\mathcal{T}^*\mathcal{E}$.
\end{theorem}

\begin{proof}
  It follows from Eq.~\eqref{eq:3} that
  \begin{equation*}
    \ell_F(H(p)) - H^*(\ell_F(p)) = \nabla(F, p),
  \end{equation*}
  on the ambient~$J^\infty(\hat{\pi})$, where
  $\nabla\colon \kappa\times\hat{P}\to P$ is a bi-differential operator in
  total derivatives. Then
  \begin{equation*}
    H_p = \eval{-\frac{1}{2}\nabla^{*_1}(p, p)}_{\mathcal{E}}
  \end{equation*}
  is the desired element. Here~$\nabla^{*_1}$ denotes the adjoint to~$\nabla$
  with respect to the first argument.
\end{proof}

\begin{remark}
  Due to Eq.~\eqref{eq:2}, the operator~$H$ sends symmetries of the
  equation~$\mathcal{E}$ to its cosymmetries. Then, for two conservation
  laws~$\omega$ and~$\omega'$, we can define their
  bracket~$\{\omega, \omega'\}_H$:
  \begin{equation*}
    \mathcal{L}_{\delta(\omega)}(\delta(\omega')) = \delta(\{\omega,
    \omega'\}_H), 
  \end{equation*}
  where~$\mathcal{L}$ is the Lie derivative and
  $\delta\colon E_1^{0,n-1}\to E_1^{1,n-1}$ is the differential in the first
  term of Vinogradov's $\mathcal{C}$-spectral sequence,~\cite{Vin-C-spec-DAN}.
  \hfill $\diamond$
\end{remark}

\begin{definition}
  For the two bivectors~$H$ and~$H'$, define their Schouten
  bracket  $\ldb H, H'\rdb$ by
  \begin{equation*}
    [\Ev_{\phi(H)}, \Ev_{\phi(H')}] = \Ev_{\ldb H, H'\rdb},
  \end{equation*}
  where~$[\cdot\,, \cdot]$ is the commutator of odd vector fields
  (anticommutator). A bivector is a Poisson structure if $\ldb H, H\rdb = 0$,
  i.e., if~$\Ev_{\phi(H)}$ is a nilpotent vector field. Two structures are
  compatible if~$\ldb H, H'\rdb = 0$.
\end{definition}

\begin{theorem}
  If~$H$ is a Poisson structure, then~$\{\cdot\,, \cdot\}_H$ is a Poisson
  bracket on the space of conservation laws, i.e., the bracket is
  skew-symmetric and enjoys the Jacobi identity.
\hfill $\Box$  
\end{theorem}

%=====================================================================

\section{Symmetries, cosymmetries, and coverings}

The Lie algebra of infinitesimal contact symmetries of Eq. \eqref{main_eq} is generated by the functions 
\[
\begin{array}{rcl}
\xi_1 &=& t\,u_t+u,
\\
\xi_2 &=& u_t,
\\
\xi_3 &=& x\,u_x-y\,u_y,
\\
\xi_4 &=& t\,x\,u_x-t\,y\,u-x\,y,
\\
\xi_5 &=& x\,u_x+y\,u_y-2\,u,
\\
\Xi_1 (A_1)&=& A_1(t)\,u_x-A_1^{\prime}(t)\,y,
\\
\Xi_2 (A_2)&=& A_2(t)\,u_y+A_2^{\prime}(t)\,x,
\\
\Xi_3 (A_3)&=& A_3(t),
\end{array}
\]
where $A_i(t)$ and $B_i(t)$ below are arbitrary smooth functions. The cosymmetries that depend on the second order jets have the form
\begin{equation}
\begin{array}{rcl}
\psi_1 &=& u,
\\
\psi_2 &=& x\,y,
\\
\Psi_1 (B_1)&=& B_1(t)\,x,
\\
\Psi_2 (B_2)&=& B_2(t)\,y,
\\
\Psi_3 (B_3)&=& B_3(t),
\\
\Upsilon&=& G(\D u),
\end{array}
\label{cosymmetries}
\end{equation}
where $B_i$ and $G$ are arbitrary smooth function of their arguments.

The tangent covering and the cotangent coverings of Eq. \eqref{main_eq} are obtained by appending the equations
\begin{equation}
\D q_t = \J (q, \D u)+\J(u, \D q)
\label{tangent_covering}
\end{equation}and
\begin{equation}
\D (p_t -\J(u,p)) + \J (\D u, p)=0,
\label{cotangent_covering}
\end{equation}
respectively.

Eq.~\eqref{main_eq} admits, in particular, the following family of
differential coverings~\cite{Morozov2024}
\begin{equation}
\left\{
\begin{array}{rcl}
s_t &=& \J(u, s)+\varepsilon\,\E(u),
\\
\J(\D u, s) &=& \lambda + \mu\,\D u-\varepsilon\,\E(\D u),
\end{array}
\right.
\label{s_covering}
\end{equation}
where $\E(\phi) = x\,D_x(\phi)+y\,D_y(\phi)-2\,\phi$ and
$\varepsilon, \lambda, \mu \in \mathbb{R}$.

%=====================================================================

\section{The local Hamiltonian structure and its inverse}

We have the following observation:
\begin{theorem}
  Let $p$ be a solution to Eq.~\eqref{cotangent_covering}, then
  $q = p_t - \J(u, p)$ is a solution to Eq.~\eqref{tangent_covering}.  \hfill
  $\Box$
\end{theorem}

Substituting for $p_t=\J(u,p) +q$ into Eq.~\eqref{cotangent_covering} gives
the system
\begin{equation}
\left\{
\begin{array}{rcl}
p_t &=& \J(u, p)+q,
\\
\J(\D u, p) &=& -\D q.
\end{array}
\right.
\label{Backlund_transformation}
\end{equation}
The compatibilty conditions of this system with respect to $p$ and $q$
coincide with Eqs.~\eqref{tangent_covering} and~\eqref{cotangent_covering},
respectively, so we have the following result:
\begin{theorem}
  System~\eqref{Backlund_transformation} defines a B{\"a}cklund transformation
  between the tangent covering and the cotangent covering of
  Eq.~\eqref{main_eq}.  \hfill $\Box$
\end{theorem}

Moreover, the following statement holds.

\begin{theorem}
\label{Hamiltonian_operator_theorem}
The mapping $H =D_t - u_x\,D_y +u_y\,D_x$, $H\colon p \mapsto q=p_t-\J(u,p)$
is a Hamiltonian operator on Eq.~\eqref{main_eq}.
\end{theorem}
\noindent
Proof. For $F = \D u_t - \J (u, \D u)$ we have
$\nabla (F,p)=-\J(F,p)=-(p_y\,D_x-p_x\,D_y)(F)$ and
$\nabla^{{*}_1}(F,p)=-(p_x\,D_y -p_y\,D_x)(F)$. Hence, we get
$\nabla^{{*}_1}(p,p)=-2\,p_xp_y$.  Since an evolutionary vector field commutes
with the total derivatives, we need to check that
$\mathbf{E}_{\varphi(H)}^2 (u) =0$ and $\mathbf{E}_{\varphi(H)}^2(p)=0$ for
the generator
$\varphi(H) =\displaystyle{\left(H(p),
    -\textfrac{1}{2}\,\nabla^{{*}_1}(p,p)\right)} = (p_t-\J(u,p), p_xp_y)$ of
a symmetry of Eqs. \eqref{main_eq}, \eqref{cotangent_covering}.  We have
$\mathbf{E}_{\varphi(H)}^2 (u) =\mathbf{E}_{\varphi(H)} (H(p)) = D_t(p_x
p_y)-D_x(H(p)) p_y-u_x\,D_y(p_x p_y) +D_y(H(p)) p_x +u_y\,D_x(p_x p_y)=0$ and
$\mathbf{E}_{\varphi(H)}^2 (p) =\mathbf{E}_{\varphi(H)} (p_x p_y) = D_x(p_x
p_y) p_y - p_x D_y(p_x p_y)=0$.  \hfill $\Box$

\begin{remark}
  The mapping $H^{-1} \colon q \mapsto p$ is a nonlocal symplectic structure
  on Eq. \eqref{main_eq}.  \hfill $\diamond$
\end{remark}

\begin{remark}
  The mapping $H$ also defines a Hamiltonian operator for
  Eq.~\eqref{Euler_eq}. The proof of this statement follows the proof of
  Theorem~\ref{Hamiltonian_operator_theorem}.  \hfill $\diamond$
\end{remark}

The action of the Hamiltonian operator $H$ on the cosymmetries
\eqref{cosymmetries} is given by the formulas
\[
\begin{array}{rcl}
H(\psi_1) &=& \xi_2,
\\
H(\psi_2) &=& -\xi_3,
\\
H(\Psi_1(A_1)) &=& \Xi_2(A_1),
\\
H(\Psi_1(A_2)) &=& -\Xi_1(A_2),
\\
H(\Psi_3(A_3)) &=& \Xi_3(\tilde{A}_3),
\\
H(\Upsilon) &=& 0,
\end{array}
\]
where $\tilde{A_3}$ is an antiderivative of $A_3$. These formulas imply the local expressions for the action of the inverse operator
$H^{-1}$ on the symmetries $\xi_2$, $\xi_3$, and $\Xi_i(A)$, while the action on 
$\xi_1$,  $\xi_4$, and $\xi_5$ has the form
\begin{equation*}
  \begin{array}{rcl}
    H^{-1}(\xi_1) &=&-s_1+t\,u,
    \\
    H^{-1}(\xi_4) &=& -s_2-t\,x\,y,
    \\
    H^{-1}(\xi_5) &=&-s_3,
  \end{array}
\end{equation*}
where $s_1$, $s_2$, and $s_3$ are solutions to Eqs. \eqref{s_covering} with
$(\varepsilon,\lambda,\mu)=(0,0,1)$, $(\varepsilon,\lambda,\mu)=(0,-1,0)$, and
$(\varepsilon,\lambda,\mu)=(-1,0,2)$, respectively.

%=====================================================================

\section{A nonlocal symplectic structure and
  recursion operator for
  co\-symmetries}

We linearize the covering \eqref{s_covering}, that is, replace  $u \mapsto u+\tau\,q$, $s \mapsto s+\tau\,r$, differentiate w.r.t. $\tau$
and then put $\tau = 0$. This yields  the system
\begin{equation}
\left\{
\begin{array}{rcl}
r_t &=& \J(u, r)+\J(q,s)+\varepsilon\,\E(q),
\\
\J(\D u, r) &=&-\J(\D q,s)+\mu\,\D q - \varepsilon\,\E(\D q),
\end{array}
\right.
\label{r_covering}
\end{equation}
which is compatible iff $u$ is a solution to Eq. \eqref{main_eq} and $q$ is a solution to Eq. \eqref{tangent_covering}.

\begin{theorem}
When $\lambda =0$ and $\mu = - 2\,\varepsilon$, a solution $s$ to Eqs. \eqref{s_covering} and the function
\begin{equation}
\EuScript{R}(q,r,s)=r -\frac{1}{2}\,\D q \cdot \left(\frac{s_x-\varepsilon\,y}{\D u_x} + \frac{s_y+\varepsilon\,x}{\D u_y}\right)
\label{nonlocal_symplectic_structure}
\end{equation}
are solutions to Eq. \eqref{cotangent_covering}.
\hfill $\Box$
\end{theorem}

\begin{remark}
When $s$ is fixed, the map $q \mapsto \EuScript{R}(q,r,s)$ is a nonlocal symplectic structure.
When $q$ is fixed, the map $s \mapsto \EuScript{R}(q,r,s)$ is a nonlocal recursion operator for cosymmetries.
Furthermore,  we have  
$
H(\EuScript{R}(q,r,s)) = - 2\,\varepsilon\,q. 
$
\hfill $\diamond$ 
\end{remark}

\begin{remark} For the Euler Eq.~\eqref{Euler_eq}, we replace $\D$ with
  $\Delta$ in \eqref{s_covering} and ~\eqref{r_covering} and modify formula
  \eqref{nonlocal_symplectic_structure} as follows:
  \begin{equation*}
    \EuScript{R}(q,r,s)=r -\frac{\Delta q}{(\Delta u_x)^2+(\Delta u_y)^2} \cdot 
    \left((s_x-\varepsilon\,y) \,\Delta u_x + (s_y+\varepsilon\,x)\,\Delta
      u_y\right). 
  \end{equation*}
  \hfill $\diamond$
\end{remark}

\section{A Hamiltonian operator for the Euler equation on a two-dimen\-sional
  Riemannian manifold}

The dynamics of inviscid incompressible fluid on a Riemannian manifold is
described by the Euler equation in vorticity form~\eqref{Euler_eq}, where the
Laplace operator
\begin{equation}
  \Delta u = \mathrm{e}^{-h}\, (u_{xx}+u_{yy})
  \label{isothermal_Laplacian}
\end{equation}
and the Jacobi bracket 
\begin{equation}
  \mathrm{J}(a,b) = \mathrm{e}^{-h}\,(a_x\,b_y-a_y\,b_x)
  \label{isothermal_Jacobi_bracket}
\end{equation}
are written in local isothermal coordinates $(x,y)$ such that the Riemannian
metric has the form $\mathrm{e}^{-h}\,(dx^2+dy^2)$ for a smooth function
$h=h(x,y)$,~\cite{Morozov2024b}. The proof of the following statement repeats
the proof of Theorem~\ref{Hamiltonian_operator_theorem}.
  
\begin{theorem}\label{thm:hamilt-oper-euler}
  The mapping $H =D_t - \mathrm{e}^{-h}\,(u_x\,D_y -u_y\,D_x)$ is a
  Hamiltonian operator on Eq.~\eqref{main_eq} with
  \eqref{isothermal_Laplacian}, \eqref{isothermal_Jacobi_bracket} and the
  generating function
  $\phi(H)=(p_t-\mathrm{e}^{-h}\,(u_x\,p_y -u_y\,p_x),\mathrm{e}^{-h}\,p_x
  p_y)$ of a symmetry of the cotangent covering.  \hfill $\Box$
\end{theorem}

%======================================================================

\section{Acknowledgments}

Computations  were done using the {\sc Jets} software \cite{Jets}.

%======================================================================

\bibliographystyle{amsplain}

\begin{thebibliography}{10}


\bibitem{ArnoldKhesin} V.I.~Arnold, B.A.~Khesin. {\it Topological Methods in
    Hydrodynamics}. Sprin\-ger, 1998

\bibitem{Jets} H.~Baran, M.~Marvan. {\it Jets. A software for differential
    calculus on jet spaces and diffieties}.  {\tt http://jets.math.slu.cz}

\bibitem{KerstenKrasilshchikVerbovetskyVitolo2009} P. Kersten,
  I.S.~Krasil$^{\prime}$shchik, A.M.~Verbovetsky, R.~Vitolo.  Hamiltonian
  structures for general PDEs.  Differential Equations --- Geometry,
  Symmetries and Integrability: The Abel Symposium 2008 (B.~Kruglikov,
  V.~Lychagin, and E.~Straume, eds.), Abel Symposia 5, Springer, 2009,
  pp.~187--198

\bibitem{KrasilshchikVerbovetsky2011} %
  J.~Krasil'shchik, A.~Verbovetsky.  Geometry of jet spaces and integrable
  systems. J.~Geom.~Phys. {\bf 61} (2011), pp.~1633--1674

\bibitem{KrasilshchikVerbovetskyVitolo2017} %
  J.~Krasil$^{\prime}$shchik, A.~Verbovetsky, R.~Vitolo.  {\it The Symbolic
    Computation of Integrability Structures for Partial Differential
    Equations}.  Springer 2017

\bibitem{KrasilshchikVinogradov1984} %
  I.S.~Krasil$^{\prime}$shchik, A.M. Vinogradov. Nonlocal symmetries and the
  theory of coverings.  Acta Appl.\ Math. {\bf 2} (1984), pp.~79--86

\bibitem{KrasilshchikVinogradov1989} %
  I.S.~Krasil$^{\prime}$shchik, A.M.~Vinogradov. Nonlocal trends in the
  geometry of differential equations: sym\-met\-ri\-es, conservation laws, and
  B\"{a}cklund transformations.  Acta Appl.\ Math. {\bf 15} (1989),
  pp.~161--209

\bibitem{KuznetsovMikhailov1980}
E.A. Kuznetsov, A.V. Mikhailov. On the topological meaning of canonical Clebsch variables. Phys. Lett. {\bf 77A} (1980), 37--38

\bibitem{LandauLifshits6}
  L.D.~Landau, E.M.~Lifshitz. {\it Course of
    Theoretical Physics. Vol. 6. Fluid Mechanics}. 2${}^{\mathrm{nd}}$ English
  ed., revised.  Pergamon Press, Oxford, 1987

\bibitem{Morozov2024}
  O.I. Morozov. Extensions of the symmetry algebra and Lax
  representations for the two-di\-men\-si\-o\-nal Euler equation.
  J.\ Geom.\ Phys. {\bf 202} (2024), 105233

\bibitem{Morozov2024b}
  O.I.~Morozov. Lax representations for the Euler ideal
  hydrodynamics equation in vorticity form on a two-dimensional Riemannian
  manifold.  J.\ Geom.\ Phys. {\bf 206} (2024), 105299

\bibitem{Olver1982}
  P.J.~Olver. A nonlinear Hamiltonian structure for the
  Euler equations. J.\ Math.\ Anal.\ Appl. {\bf 89} (1982), pp.~233--250

\bibitem{Vin-C-spec-DAN}
  A.M.~Vinogradov, A spectral sequence associated with
  a nonlinear differential equation and algebro-geometric foundations of
  Lagrangian field theory with constraints, Soviet Math.\ Dokl.\ {\bf 19}
  (1978), pp.~144--148

\bibitem{VK1999} %
  A.M.~Vinogradov, I.S.~Krasil$^{\prime}$shchik (eds.) {\it Symmetries and
    Conservation Laws for Differential Equations of Mathematical Physics} [in
  Russian], Moscow: Factorial, 2005; English transl. prev. ed.:
  I.S.~Krasil$^{\prime}$shchik, A.M.~Vinogradov (eds.) {\it Symmetries and
    Conservation Laws for Differential Equations of Mathematical
    Physics}. Transl.\ Math.\ Monogr., {\bf 182}, Amer. Math.\ Soc.,
  Providence, RI, 1999

\end{thebibliography}

\end{document}